\documentclass{article}

\usepackage{fullpage}
\usepackage{color}
\usepackage{amssymb}
\usepackage{amsthm}
\usepackage{amsmath}
\usepackage[ruled,vlined]{algorithm2e}
\usepackage{authblk}

\newcommand{\bx}{{\bf{x}}}
\newcommand{\by}{{\bf{y}}}

\begin{document}

\title{Online Vertex Cover and Matching:\\ Beating the Greedy Algorithm}


\author[1]{Yajun Wang\thanks{yajunw@microsoft.com}}
\author[2]{Sam Chiu-wai Wong\thanks{samwong@mit.edu. Part of the work is done when the author was visiting Microsoft Research Asia.}}
\affil[1]{Microsoft Research Asia, Beijing, P.~R.~China}
\affil[2]{Massachusetts Institute of Technology}


\maketitle
\begin{abstract}


In this paper, we {\em explicitly} study the online vertex cover problem, which is a natural generalization of the well-studied ski-rental problem. In the online vertex cover problem, we are required to maintain a monotone vertex cover in a graph whose vertices arrive online. When a vertex arrives, all its incident edges to previously arrived vertices are revealed to the algorithm. For bipartite graphs with the left vertices offline (i.e. all of the left vertices arrive first before any right vertex), there are algorithms achieving  the optimal competitive ratio of $\frac{1}{1-1/e}\approx 1.582$.

Our first result is a new optimal {\em water-filling} algorithm for this case. One major ingredient of our result is a new charging-based analysis, which can be generalized to attack the
online {\em fractional} vertex cover problem in general graphs. The main contribution of this paper is a 1.901-competitive algorithm for this problem. 
When the underlying graph is bipartite, our fractional solution can be rounded to an integral solution. In other words, we can obtain a vertex cover with expected size at most $1.901$ of the optimal vertex cover in bipartite graphs.

The next major result is a primal-dual analysis of our algorithm for the online fractional vertex cover problem in general graphs.  
This is more than just yet another analysis as it implies the dual result of a 0.526-competitive algorithm for online {\em fractional} matching in general graphs. Notice that both problems admit a well-known 2-competitive greedy algorithm. Our result in this paper is the first successful attempt to beat the greedy algorithm for these two problems.

Our result on the online matching problem significantly generalizes the traditional online bipartite graph matching problem, where vertices from only one side of the bipartite graph arrive online. In particular, our algorithm improves upon the result of the fractional version of the online edge-selection problem in Blum et. al. (JACM '06), where only the greedy algorithm with competitive ratio $1/2$ is previously known.

On the hardness side, we show that no randomized online algorithm can achieve a competitive ratio better than 1.753 and 0.625 for the online fractional vertex cover problem and the online fractional matching problem respectively, even for bipartite graphs.
\end{abstract}

\newtheorem{lemma}{Lemma}
\newtheorem{theorem}{Theorem}
\newtheorem{corollary}{Corollary}
\newtheorem{proposition}{Proposition}

\section{Introduction}

In this paper, we study the online vertex cover problem in bipartite and general graphs. Given a graph $G=(V,E)$, $C\subseteq V$ is a vertex cover of $G$ if all edges in $E$ are incident to $C$. In the online setting, the vertices of $V$ arrive one at a time. When a vertex arrives, its edges incident to the previously arrived neighbors are revealed. 
We are required to maintain a {\em monotone} vertex cover for the revealed subgraph at all time. In particular, no vertices can be removed from the cover once added. The objective is to minimize the size of the final vertex cover.

Our study of the online vertex cover problem is motivated by two apparently unrelated lines of research in the literature, namely ski rental and online bipartite matching. 

\paragraph{Online bipartite matching.}
The online bipartite matching problem has been intensively studied over the past decade. An instance of this problem specifies a bipartite graph $G=(L,R,E)$ in which the set of left vertices $L$ is known in advance, while the set of right vertices $R$ and edges $E$ are revealed over time. An algorithm maintains a monotone matching that is empty initially. At each step, an online vertex $v\in R$ arrives and all of its incident edges are revealed. An algorithm must immediately and irrevocably decide if $v$ should be matched to a previously unmatched vertex in $N(v)$. The objective is to maximize the size of the matching found at the end.

This problem and almost all of its variants studied in the literature share the common feature that vertices of only one side of the bipartite graph arrive online. While this property indeed holds in many applications, it does not necessarily reflect the reality in general. We exemplify this by the following application:
\begin{itemize}

\item {\bf Online market clearing.} In a commodity market, buyers and sellers are represented by the left and right vertices. An edge between a buyer and seller indicates that the price that the buyer is willing offer is higher than the price at which the seller is willing to take. The objective is to maximize the number of trades, or the size of the matching. In this problem, both the buyers and sellers arrive and leave online continuously.

\end{itemize}
Thus a more general model of online bipartite matching is to allow all vertices to be online. In this paper, we obtain the first non-trivial algorithm for the fractional version of this generalization. Our algorithm is 0.526-competitive and, in fact, also works in general graphs.

\paragraph{Ski rental and online bipartite vertex cover.}
The ski rental problem is perhaps one of the most studied online problems. Recall that in this problem, a skier goes on a ski trip for $N$ days but has no information about $N$. On each day he has the choice of renting the ski equipment for 1 dollar or buying it for $B>1$ dollars. His goal is to minimize the amount of money spent.

We consider the online bipartite vertex cover problem, which is a generalization of ski rental. The setting of this problem is exactly identical to that of online bipartite matching except that the task is to maintain a monotone vertex cover instead. Ski rental can be reduced to online bipartite vertex cover via a complete bipartite graph with $B$ left vertices and $N$ right vertices. One may view this problem as ski rental with a combinatorial structure imposed.
We show that the optimal competitive ratio of online bipartite vertex cover is $\frac{1}{1-1/e}$. In other words, we still have the same performance guarantee even though the online bipartite vertex cover problem is considerably more general than the ski rental problem.

\paragraph{The connection.}
Recall that bipartite matching and vertex cover are dual of each other in the offline setting. It turns out that the analysis of an algorithm for online bipartite fractional matching in~\cite{Buchbinder2007} implies an optimal algorithm for online bipartite vertex cover. On the other hand, online bipartite vertex cover generalizes ski rental. This connection is especially interesting because online bipartite matching does not generalize ski rental but is the dual of its generalization\footnote{Coincidentally, the first papers on online bipartite matching and ski rental were both published in 1990 but to our knowledge, their connection was not realized, or at least explicitly stated.}.

\paragraph{The greedy algorithm.}
There is a simple well-known greedy algorithm for online matching and vertex cover in general graphs. As each vertex arrives, we match it to an arbitrary unmatched neighbor (if any) and put both of them into the vertex cover. It is easy to show that this algorithm is $1/2$-competitive for online matching and $2$-competitive for online vertex cover.

The greedy algorithm for the vertex cover problem is optimal assuming the Unique Game Conjecture even in the offline setting~\cite{Khot2008}. Thus there is no hope of doing better than 2 if we strict ourselves to integral vertex covers in general graphs.  For the other problems studied in this paper, e.g. matching and vertex cover in bipartite graphs and matching in general graphs, no known algorithm beats the greedy algorithm in the online setting.

We present the first successful attempt in breaking the barrier of 2 (or 1/2) achieved by the greedy algorithm. In the fractional setting, our algorithm is $1.901$-competitive (against the minimum fractional cover) for online vertex cover and $\frac{1}{1.901}\approx 0.526$-competitive (against the maximum fractional matching) for online matching in general graphs. It is possible to convert the fractional algorithm to a randomized integral algorithm for online vertex cover in bipartite graphs. On the other hand, it is not clear whether it is possible to round our algorithm or its variants for online matching in either bipartite graphs or general graphs.

We stress that the fractional setting is still of interest for two reasons:

\begin{itemize}
\item As well-articulated in~\cite{Buchbinder2007}, some commodities are divisible and hence should be modeled as fractional matchings. In fact, for divisible commodities one would even prefer a fractional matching assignment since the maximum fractional matching may be larger than the maximum integral matching in general graphs. Thus a $c$-competitive algorithm against fractional matching would be preferable to a $c$-competitive algorithm against integral matching.
\item Our 0.526-competitive algorithm for fractional matching suggests that it may be possible to beat the greedy algorithm for online integral matching in the oblivious adversarial model.
\end{itemize}

%
%
\subsection{Our results and techniques}

Our algorithms rely on a charging-based algorithmic framework for online vertex cover-related problems. The following results on vertex cover were obtained using this method:
\begin{itemize}
\item A new optimal $\frac{1}{1-1/e}$-competitive algorithm for online bipartite vertex cover\footnote{A similar algorithm is implied by the analysis of the algorithm for online bipartite fractional matching in~\cite{Buchbinder2007}.}.
\item A 1.901-competitive algorithm for online fractional vertex cover in general graphs.
\end{itemize}

We stress that the fact that our result holds only for the \emph{fractional} version of online vertex cover in general graphs is reasonable. In fact, even in the offline setting, the best known approximation algorithm for minimum vertex cover is just the simple 2-approximate greedy algorithm. Getting anything better than 2 would disprove the Unique Game Conjecture even in the offline setting~\cite{Khot2008} and have profound implications to the theory of approximability.

Our algorithms can also be analyzed in the prime-dual framework~\cite{Buchbinder2007}. As by-products, we obtain dual results on the maximum matching as follows:
\begin{itemize}
\item A 0.526-competitive algorithm for online fractional matching in general graphs. This improves the result on the online edge-selection problem studied in~\cite{Blum2006}.
\end{itemize}

All of these results also hold in the vertex-weighted setting (for vertex cover) and the b-matching setting.
Section~\ref{sec:rounding} explains how to convert essentially any algorithm for online {\em fractional} vertex cover to an algorithm for online {\em integral} vertex cover in the case of bipartite graphs with the same (expected) performance. 



On the hardness side, we establish the following lower bound (for vertex cover) and upper bound (for matching) on the competitive ratios. Notice that these bounds also apply to the integral version of the problems.
\begin{itemize}
\item A lower bound of $1+\sqrt{\frac{1}{2}\left(1+\frac{1}{e^2}\right)}\approx 1.753$ for the  online {\em fractional} vertex cover problem in bipartite graphs.
\item An upper bound of 0.625 for the online {\em fractional} matching problem in bipartite graphs. 
\end{itemize}

\paragraph{Main ingredients.}
Our result is based on a novel charging-based analysis of a new {\em water-filling} algorithm for the online bipartite vertex cover problem. In the {\em water-filling} algorithm, for each online vertex, we are allowed to use water of amount at most $\frac{1}{1-1/e}$ to cover the new edges. (Recall that in the original {\em water-filling} algorithm for matching, the amount of water is at most $1$.) In our charging scheme, for an online vertex in the optimal cover, we charge all the water used in processing this vertex to itself. For an online vertex not in the optimal cover, we charge the water spent on the online vertex to its neighbors, which must be in the optimal cover. In particular, in the bipartite graph case with one-sided online vertices, an online vertex in the optimal cover will take care of the cost processing itself wherea an offline vertex in the optimal cover is responsible for the charge from its online neighbors. 

In generalizing the charging scheme to the two-sided online bipartite and the general graph cases, a vertex must take care of both the cost in processing itself and the charges received from future neighbors. In such generalizations, we cannot use a fixed amount of water in processing each vertex. A key insight behind our algorithm is that the amount of water used should be related to the actual final water level. In other words, for a final water level $y$, the amount of water used should be $f(y)$ for some allocation function $f(\cdot)$. By extending our previous charging scheme, the competitive ratio of our new water-filling algorithm for the online fractional vertex cover problem in general graphs will be a function of $f(\cdot)$. We also derive the constraints which $f(\cdot)$ must satisfy in order to make the analysis work.

As a result, we are left with a non-conventional minimax optimization problem. (See Eqn.(\ref{eqn:opt}).) The most exciting part, however, is that we can actually solve this optimization problem {\em optimally}.\footnote{Our solution is optimal in {\em our framework}. It may not be optimal for the online fractional vertex cover problem. } The optimal allocation function in Theorem~\ref{thm:vcgeneral} implies a competitive ratio of $1.901$ for the online fractional vertex cover problem in general graphs.
Our primal-dual analysis for the online fractional matching problem in general graphs is obtained by reverse-engineering the charging-based analysis.

\paragraph{Remark:} In retrospect, it may be much harder to directly develop a water-filling algorithm for online matching in general graphs. 
Firstly, it may take some work to realize that the amount of water used should be variable rather than 1 as in online bipartite matching. On the contrary, in vertex cover, the amount of water is already variable even for the basic one-sided online bipartite vertex cover. Secondly, to analyze a water-filling algorithm on the matching, one has to optimize over the allocation function, which specifies the total amount of water used as a function of the water level, and another function which updates the potentials of the dual variables. As a consequence, the competitive ratio would be an optimization problem involving {\em two variable functions}! In fact, if we reverse engineer a water-filling algorithm on the matching from our solution, the corresponding allocation function does not have a known closed form.
Thus our charging-based analysis for online vertex cover is a critical step in developing the algorithms. Our starting point, the vertex cover, turns out to be a surprising blessing.

\subsection{Previous work}
There are three lines of research related to our work. The first two categories discussed below are particularly relevant.

\paragraph{Online matching.}
The online bipartite matching problem was first studied in the seminal paper by Karp et al.~\cite{Karp1990}. They gave an optimal $1-1/e$-competitive algorithm. Subsequent works studied its variants such as $b$-matching~\cite{kalyanasundaram2000optimal}, vertex weighted version~\cite{Aggarwal2011,devanurrandomized}, adwords~\cite{Buchbinder2007,DevenurH09,Mehta2007,devanurrandomized, devanur2012online,goel2008online,Aggarwal2011} and online market clearing~\cite{Blum2006}. Water-filling algorithms have been used for a few variants of the online bipartite matching problem (e.g. ~\cite{kalyanasundaram2000optimal,Buchbinder2007}).

Another line of research studies the problem under more relaxed adversarial models by assuming certain inherent randomness in the inputs~\cite{Feldman2009,Manshadi2011,Mahdian2011, Karande2011}. Online matching for general graphs have been studied under similar stochastic models~\cite{bansal2010lp}. To our knowledge, there is no result on this problem in the more restricted adversarial models other than the well-known $1/2$-competitive greedy algorithm, even for just bipartite graphs with vertices from both sides arriving online~\cite{Blum2006}.

Analyzing greedy algorithms for maximum matching in the offline setting is another related research area. Aronson et al.~\cite{Aronson1995} showed that a randomized greedy algorithm is a $\frac{1}{2}+\frac{1}{400,000}$-approximation. The factor was recently improved to $\frac{1}{2}+\frac{1}{256}$~\cite{poloczek12}. A new greedy algorithm with better ratio was presented in~\cite{goel12}. Our 0.526-competitive algorithm for online fractional matching complements these results. 

\paragraph{Ski rental.} The ski rental problem was first studied in~\cite{karlin1988competitive}. Karlin et al. gave an optimal $\frac{1}{1-1/e}$-competitive algorithm in the oblivious adversarial model~\cite{Karlin1994}. There are many generalizations of ski rental. Of particular relevance are multislope ski rental~\cite{Lotker2008}  and TCP acknowledgment~\cite{Karlin2001}, where the competitive ratio $\frac{1}{1-1/e}$ is still achievable. 
The online vertex-weighted bipartite vertex cover problem presented in this paper is also of this nature and, in fact, further generalizes multislope ski rental, as shown in Appendix~\ref{sec:multislope}.

\paragraph{Online covering.} Another line of related research deals with online integral and fractional covering programs of the form $\min\{ cx\mid Ax\geq 1,0\leq x\leq u\}$, where $A\geq 0,u\geq 0$, and the constraints $Ax\geq 1$ arrive one after another~\cite{Buchbinder2009}. Our online vertex cover problem also falls under this category. The key difference is that the online covering problems are so general that the optimal competitive ratios are usually not constant but logarithmic in some parameters of the input.

Finally, online vertex cover for general graphs was studied by Demange et al.~\cite{Demange2005} in a model substantially different from ours. Their competitive ratios are characterized by the maximum degree of the graph.

\section{Preliminaries}

Given $G=(V,E)$, a vertex cover of $G$ is a subset of vertices $C\subseteq V$ such that for each edge $(u,v) \in E$, $C\cap \{u,v\} \neq \emptyset$. A matching of $G$ is a subset of edges $M\subseteq E$ such that each vertex $v \in V$ is incident to at most one edge in $M$.

$\by\in [0,1]^V$ is a fractional vertex cover if for any edge $(u,v)\in E$, $y_u+y_v \geq 1$. We call $y_v$ the {\em potential} of $v$.  $\bx\in [0,1]^E$ is a fractional matching if for each vertex $u\in V$, $\sum_{v\in N(u)} x_{uv} \leq 1$. It is well-known that vertex cover and matching are dual of each other.

{\bf LPs for fractional vertex cover and matching.}

\begin{center}
\begin{tabular}{ | r l | r l | }
\hline
Primal (Matching): & & Dual (Vertex Cover): &\\
 & $\max\sum_{e\in E} x_e$ &  & $\min\sum_{v\in V}y_v$ \\
s.t. & $ x_v:=\sum_{u\in N(v)} x_{uv}\leq 1,\,\forall v\in V$ & s.t. & $y_u+y_v\geq 1,\,\forall (u,v)\in E$ \\
&  $x\geq 0$ & & $y\geq 0$ \\
\hline
\end{tabular}
\end{center}

In this paper, the matching and vertex cover LPs are called the primal and dual LPs, respectively. By weak duality, we have $$\sum_{e\in E} x_e\leq\sum_{v\in V}y_v $$ for any feasible fractional matching $\bx$ and vertex cover $\by$.

%
%
%

{\bf Competitive analysis.}
We adopt the competitive analysis framework to measure the performance of online algorithms. The size of the vertex cover  (or matching) found by an algorithm is compared against the (offline) optimal solution, in the worst case.

An algorithm, possibly randomized, is said to be \textit{$c$-competitive} if for any instance, the size of the solution $ALG$ found by the algorithm and the size of the optimal solution $OPT$ satisfy 
\[
\mathbb{E}[ALG] \leq c\cdot OPT\,\mathrm{ or }\,\mathbb{E}[ALG]\geq  c\cdot OPT
\]
depending on whether the optimization is a minimization or maximization problem. The constant $c$ is called the \textit{competitive ratio}.

A few different adversarial models have been considered in the literature. In this paper, we focus on the {\em oblivious adversarial} model, in which the adversary must specify the input once-and-for-all at the beginning and is not given access to the randomness used by the algorithm. 

\subsection{Algorithms for online vertex cover and matching}
In the online setting, the vertices of $G$ arrive one at a time in an order determined by the adversary. When an online vertex $v$ arrives, all of its edges incident to the {\em previously arrived} vertices are revealed. We denote the set of arrived vertices by $T\subset V$ and $G(T)$ is the subgraph of $G$ induced by $T$.



An algorithm for online integral matching maintains a monotone matching $M$. As each vertex $v$ arrives, it must decide if $(u,v)$ should be added to $M$ for some previously unmatched $u\in N(v)\cap T$, where $N(v)$ is neighbors of $v$ in $G$. No edge can be removed from $M$. The objective is to maximize the size of the final matching $M$. For online fractional matching, a fractional matching $\bx$ for $G(T)$ is maintained and at each step, $x_{uv}$ must be initialized for $u\in N(v)\cap T$ so that $\bx$ remains a fractional matching. The objective is to maximize the final $\sum_{e\in E} x_e$.

An algorithm for online integral vertex cover maintains a monotone vertex cover $C$. As each vertex $v$ arrives, it must insert a subset of $\{v\}\cup N(v)\backslash C$ into $C$ so that it remains a vertex cover. No vertex can be removed from $C$. The objective is to minimize the size of the final cover $C$. For online fractional vertex cover, a fractional vertex cover $\by$ for $G(T)$ is maintained and at each step, we must initialize $y_v$ and possibly increase some $y_u$ for $u\in T$ so that $\by$ remains a fractional vertex cover. The objective is to minimize the final $\sum_{v\in V}y_v$.

To simplify the terminology, we refer to the online vertex cover (matching) problem as the instances where all vertices in the graph arrive online. On the other hand, to be conformal with the existing terminology in the literature, we refer to the online {\em bipartite} vertex cover (matching) problem as the instances where the graph is bipartite and only one side of the vertices arrive online. This is the traditional case studied in the literature.

{\bf Weighted vertex cover and b-matching.} Our results can be generalized to cases of weighted vertex cover and b-matching.
For vertex cover, the objective function becomes $\sum_{v\in C}w_v$ (integral) or $\sum_{v\in V}w_vy_v$ (fractional), where $w_v\geq 0$ are weights on the vertices that are revealed to the algorithm when $v$ arrives.

For b-matching, the only difference is that each vertex can be matched up to $w_v\in\mathbb{N}$ times instead of just 1 (integral) or the constraint $x_v:=\sum_{u\in N(v)}x_{uv}\leq w_v$, where $w_v\geq 0$, replaces $x_v\leq 1$ (fractional). See below the LP formulation of the two problems for the fractional solution. 

{\bf LPs for fractional weighted vertex cover and b-matching.}

\begin{center}
\begin{tabular}{ | r l | r l | }
\hline
Primal: & & Dual: &\\
 & $\max\sum_{e\in E} x_e$ &  & $\min\sum_{v\in V}w_vy_v$ \\
s.t. & $ x_v:=\sum_{u\in N(v)} x_{uv}\leq w_v,\, \forall v\in V$ & s.t. & $y_u+y_v\geq 1,\, \forall (u,v)\in E$ \\
&  $\bx\geq 0$ & & $\by\geq 0$ \\
\hline
\end{tabular}
\end{center}

\subsection{Rounding fractional vertex cover in bipartite graphs}
\label{sec:rounding}
We present a rounding scheme that converts any given algorithm for online {\em fractional} vertex cover to an algorithm for online {\em integral} vertex cover in bipartite graphs~\cite{NivPersonal}.\footnote{We previously had a more complex rounding scheme. We thank Niv Buchbinder for letting us present his simple scheme.} This allows us to obtain the integral version of our results on fractional vertex cover for bipartite graphs.

Let ${\bf y}$ be the fractional vertex cover maintained by the algorithm. 
Sample $t\in [0,1]$ uniformly at random before the first online vertex arrives. Throughout the execution of the algorithm, assign $u\in L$  to the cover if $y_u\geq t$ and $v\in R$ to the cover if $y_v\geq 1-t$, where $L$ and $R$ are the left and right vertices of the graph $G$ respectively.
As $y_u$ and $y_v$ never decrease in the online algorithm, our rounding procedure guarantees that once a vertex enters the cover, it will always stay there. 

We next claim that this scheme gives a valid cover. Since $\by$ is always feasible, we have $y_u+y_v\geq 1\forall (u,v)\in E$ and hence at least one of $y_u\geq t$ and $y_v\geq 1-t$ must hold. In other words, one of $u$ and $v$ must be in the cover.
Therefore the cover obtained by applying this scheme is indeed valid and monotone, as required.

Finally, for each vertex $v$ with final potential $y_v$, the probability that $v$ is in the cover after the rounding is exactly $y_v$. Therefore, by linearity of expectation, the expected size of the integral vertex cover after the rounding is exactly $\sum_{v\in L\cup R} y_v$. Hence, this rounding scheme does not incur a loss.

\section{Online bipartite vertex cover problem}

In this section, we study the online bipartite vertex cover which is the dual of the traditional online bipartite matching problem. In this problem, the left vertices of the graph $G=(L,R,E)$ are offline and the right vertices in $R$ arrive online one at a time.
As mentioned in the introduction, online bipartite vertex cover generalizes the well known ski rental problem.
\begin{lemma}
Online bipartite vertex cover generalizes ski rental. In particular, no algorithm for online bipartite vertex cover achieves a competitive ratio better than $1+\alpha:=\frac{1}{1-1/e}$, which is the optimal ratio for ski rental~\cite{Karlin1994}.
\end{lemma}

\subsection{An optimal algorithm: $GreedyAllocation$}
\label{subsec:greedyallocation}
We present an optimal algorithm for the online vertex cover problem in bipartite graphs.  Notice that the primal-dual analysis of the previously studied water-level algorithms on the online bipartite matching problem implies an optimal algorithm for the online bipartite vertex cover problem. Our algorithm applies the {\em water level} paradigm on {\em vertex cover} instead of {\em matching}. 

This difference may appear trivial but it actually has profound consequences. In the water-filling algorithms for matching, the amount of water used is typically at most 1, i.e. the online vertex can be matched at most once. This is independent of the final water level. However, in vertex cover, we use at most $y+\alpha$ amount of water on the neighbors of the online vertex when the final water level is $y$. Our use of a general allocation function $f(\cdot)$ in the general graph case is partly inspired by this. Secondly, our new algorithm permits a novel charging-base analysis, 
%
%
which encompasses several key observations that are helpful in developing our algorithm for online vertex cover in general graphs.

To avoid repetition, we present our algorithm in the general case as Algorithm~\ref{alg:general greedy} with allocation function $f(\cdot)$. For each vertex $v$, we maintain a non-decreasing cover potential $y_v$ which is initialized to $0$.
When an online vertex $v$ arrives, the edges between $v$ and $N(v)\cap T$ are revealed. In order to cover these new edges, we must increase the potential of $v$ and its neighbors. Suppose that we set $y_v=1-y$ after processing $v$. To maintain a feasible vertex cover, we must increase any $y_u<y$ for $u\in N(v)$ to $y$. We call $y$ the {\em water level}.

The trick here lies in how $y$ is determined. We consider a simple scheme in which $y$ is related to the total potential increment of $N(v)$. More precisely, we require that the total potential increment $\sum_{u\in N(v):y_u<y}(y-y_u)$ be at most $f(y)$, where $f$ is a positive continuous function on $[0,1]$. 

For the online bipartite vertex cover problem considered in this section, the {\em allocation function} $f(y)=\alpha + y$ turns out to be an optimal choice. Another interpretation of this allocation function is that we spend at most $(1-y)+(\alpha + y)=1+\alpha$ amount of water on each online vertex. This observation will be crucial in the analysis.

\begin{algorithm}[h!]
\SetAlgoLined
\caption{$GreedyAllocation$ with allocation function $f(\cdot)$}
\label{alg:general greedy}
\KwIn{Online graph $G=(V,E)$ with offline vertices $U\subset V$}
\KwOut{A fractional vertex cover of $G$}
Initialize for each $u\in U$, $y_u = 0$\;
Let $T$ be the set of known vertices. Initialize $T=U$\;
\For{each online vertex $v$}
{
Maximize $y\le 1$, s.t., $\sum_{u\in N(v)\cap T} \max\{y-y_u,0\} \leq f(y)$\;
For each $u\in N(v)\cap T$, $y_u \leftarrow \max\{ y_u, y\}$\;
$y_v \leftarrow 1-y$\;
$T\leftarrow T\cup \{v\}$\;
}
Output $\{y_v\}$ for all $v\in V$\;
\end{algorithm}

\subsection{Analyzing $GreedyAllocation$}



Now we analyzing the performance of $GreedyAllocation$ with $f(y) = y +\alpha$ for the online bipartite vertex cover problem.
Let $C^*$ be a minimum vertex cover of $G$. Our strategy is to charge the potential increment to vertices of $C^*$ in such a way that each vertex of $C^*$ is charged at most $1+\alpha$.

Let $v$ be the current online vertex. Suppose that our algorithm sets $y_v=1-y$ for some $y$. 
Let $y_u$ be the potential of $u\in N(v)$.
We consider two cases.

\underline{Case 1:} $v\in C^*$. 
It is natural to charge the potential increment in $N(v)$ and $v$ to $v$. 
By our construction, $v$ will be charged at most $1+\alpha$.

\underline{Case 2:} $v\notin C^*$. Notice that we must have $N(v)\subseteq C^*$. In this case, vertices of $N(v)$ should be responsible for the potential $y_v=1-y$ used by $v$.
We describe how to charge $1-y$ to $N(v)$ as follows.

Intuitively, if $\sum_{u\in N(v)} (y-y_u)=f(y)=\alpha + y$, the most fair scheme should charge $\frac{1-y}{f(y)}(y-y_u)$ to $u \in N(v)$ whose potentials increase
since  the fair ``unit charge" is $\frac{1-y}{f(y)}$. 
If $\frac{1-t}{f(t)}$ is decreasing, $\frac{1-y}{f(y)}(y-y_u)$ can be upper bounded by $\int_{y_u}^y \frac{1-t}{f(t)}\mathrm{d}t$. This observation motivates the next lemma which forms the basis of all the major results in this paper.
\begin{lemma}
\label{lem:charging}
Let $f:[0,1]\longrightarrow\mathbb{R}_+$ be continuous such that $\frac{1-t}{f(t)}$ is decreasing, and $F(x)=\int_0^x  \frac{1-t}{f(t)}\mathrm{d}t$. If $\sum_{u\in X} (y-y_u)= f(y)$ for some set $X$ and $y\geq y_u$ for $u\in X$, then
$$1-y\leq \sum_{u\in X} \left(F(y)-F(y_u)\right).$$
\end{lemma}
\begin{proof}
We have the following
\begin{eqnarray*}
\sum_{u\in X}\left(F(y)-F(y_u)\right) & = & \sum_{u\in X}\int_{y_u}^y  \frac{1-t}{f(t)}\mathrm{d}t\\
             & \geq &  \sum_{u\in X}(y-y_u)\frac{1-y}{f(y)}=1-y,
\end{eqnarray*}
where the inequality above holds as $\frac{1-t}{f(t)}$ is decreasing.
\end{proof}

%

We are ready to evaluate the performance of $GreedyAllocation$. 
\begin{theorem}
\label{thm:no alternation}
$GreedyAllocation$ is $1+\alpha$-competitive and hence optimal for the online bipartite vertex cover problem.
\end{theorem}
\begin{proof}
We charge the potentials used to the vertices of the minimum cover $C^*$. Let $v$ be an online vertex. The case $v\in C^*$ is trivial as explained before.



Now consider the case $v\notin C^*$. We charge the potential spent on $u\in N(v)\subseteq C^*$ to $u$ itself. The potential spent on $v$ is $y_v = 1-y$ where $y$ is the final water level after processing $v$. 
Let $X\subset N(v)$ be the set of vertices whose potentials increase when processing $v$.
If $y =1$, we are done as no charging is necessary. If $y <1$, then we have $\sum_{u\in X} (y-y_u) = \alpha +y$, where $y_u$ is the potential of $u$ before processing $v$. We charge each vertex $u\in X$ by $G(y)-G(y_u)$. By Lemma~\ref{lem:charging}, $1-y \leq \sum_{u\in X} (G(y)-G(y_u))$, i.e., our charging is sufficient.



In summary, each online vertex of $C^*$ is responsible for $1+\alpha$ potential. On the other hand, each left vertex of $C^*$ is responsible for itself (which contributes at most 1 to $C$) as well as the incoming charges from its neighbors. For $u\in L\cap C^*$, the sum of these charges can be at most $G(1)-G(0)$ as the sum $G(y)-G(y_u)$, taken over the iterations in which $y_u$ increases, telescopes. Therefore the amount of potential charged to a left vertex is also bounded by $1+G(1)-G(0) = (\alpha+1)\ln (1+\frac{1}{\alpha}) = 1+\alpha$ since $\alpha = \frac{1}{e-1}$.
This gives our desired result.
\end{proof}

In fact, $GreedyAllocation$ can be extended to the vertex-weighted setting. To avoid diversion from the main results, we defer the proof of the following theorem to the appendix.
\begin{theorem}
$GreedyAllocation$ (modified) is $1+\alpha$-competitive and hence optimal for online vertex-weighted bipartite vertex cover.
\end{theorem}

%

\section{Online fractional vertex cover in general graphs}
The lessons learned in the last section are actually much more general. As suggested in the description of $GreedyAllocation$, we can generalize the algorithm to general graphs. However, we have to carefully design the allocation function $f(\cdot)$ to get a non-trivial competitive ratio.

Before getting into the details, we revisit the analysis in the last section to gain some insights which will be helpful to tackle the general graph version of the problem. In our charging argument, each vertex in $L\cap C^*$ is responsible for the charges from its neighbors. On the other hand, a vertex in $R\cap C^*$ is only responsible for the potential increment when processing itself. However, if both vertices in $L$ and $R$ are online, an online vertex $v\in C^*$ should be responsible for the potential used to process it when it arrives as well as the charges from future neighbors.

Let $f(x)$ be a general allocation function such that $\frac{1-t}{f(t)}$ is decreasing. Informally, if the water level when processing $v$ is $y<1$, i.e. the initial potential of $v$ is $1-y$, we use potential $f(y)$ on $v$'s neighbors and $1-y$ on $v$ itself. Afterwards, $v$ will take charges from its future neighbors. Notice that $v$'s potential will grow from $1-y$ to at most $1$. By Lemma~\ref{lem:charging}, $v$ will take charges at most $\int_{1-x}^1 \frac{1-t}{f(t)} \mathrm{d}t$. Putting the two pieces together, the total charges to each $v\in C^*$ and hence the competitive ratio are at most 
\[
\beta(f) = \max_{z\in [0,1]} 1 + f(1-z) + \int_{z}^1 \frac{1-t}{f(t)}\mathrm{d}t.
\]

We will show how to compute the optimal allocation function $f(\cdot)$ in Sec.~\ref{sec:optimization}. From now on, we will formally show that the performance of $GreedyAllocation$ in general graphs with allocation function $f(\cdot)$ is at most $\beta(f)$.

\begin{lemma}
Let $f(\cdot)$ be the allocation function.
In processing vertex $v$ in $GreedyAllocation$, we must have either $y=1$ or $\sum_{u\in N(v)} \max\{y-y_u,0\}= f(y)$.
\end{lemma}
\begin{proof}
Let $H(t)=\sum_{u\in N(v)}\max\{t-y_{u},0\}- f(t)$. Note that $H$ is continuous and $H(0)=-f(0)<0$.

Assume $y< 1$. Notice that $H(1)>0$. Otherwise, we can set $y=1$. If $H(y) <0$, then by intermediate value theorem there is some $t\in (y,1)$ for which $H(t)=0$. This contradicts the maximality of $y$. Hence $H(y)=0$, as desired.
\end{proof}

Our previous discussion implies that $GreedyAllocation$ is competitive against the minimum {\em integral} vertex cover. In fact, our algorithm is also competitive against the minimum {\em fractional} vertex cover in general graphs.
\begin{theorem}
Let $f: [0,1]\longrightarrow R^+$ be the continuous allocation function such that $\frac{1-t}{f(t)}$ is decreasing.
Let $\beta=\max_{z\in [0,1]} 1+f(1-z)+\int_z^1 \frac{1-t}{f(t)}\mathrm{d}t$ and $F(x)=\int_0^x  \frac{1-t}{f(t)}\mathrm{d}t$. 
$GreedyAllocation(f)$ is $\beta$-competitive against the optimal fractional vertex cover in general graphs.
\end{theorem}
\begin{proof}
Let $ y^*$ be the minimum fractional vertex cover. Denote by $v$ the current online vertex. Consider the following charging scheme.
\begin{itemize}
\item Charge $\left(f(y)+1-y\right)y_v^*$ to $v$.
\item Charge $\left(y-y_u+F(y)-F(y_u)\right)y_u^*$ to $u\in X$, where $X=\{ u\in N(v)\mid y_u<y\}$.
\end{itemize}
We claim that the total charges are sufficient to cover the potential increment $1-y+\sum_{u\in X}(y-y_u)$.

Observe that since $y_v^*+y_u^*\geq 1$ for all $u\in N(v)$. Since $f(y)\geq \sum_{u\in X}(y-y_u)$, we have 
\begin{eqnarray*}
f(y)y_v^*+\sum_{u\in X}(y-y_u)y_u^* & \geq & \sum_{u\in X}(y-y_u)(y_v^*+y_u^*)\\
             & \geq &  \sum_{u\in X}(y-y_u).
\end{eqnarray*}
Furthermore,
\begin{eqnarray*}
&&(1-y)y_v^*+\sum_{u\in X}\left(F(y)-F(y_u)\right)y_u^* \\
&\geq&  (1-y)y_v^*+\sum_{u\in X}\left(F(y)-F(y_u)\right)(1-y_v^*)\\
&\geq&   1-y,
\end{eqnarray*}
where the last inequality follows from Lemma~\ref{lem:charging}.

The above shows that the proposed charging scheme indeed accounts for the total potential increment. Now we bound the total charges to a vertex $v$ over the execution of the algorithm.

When $v$ arrives, $y_v$ is initialized as $1-y$ and $v$ is charged $\left(f(y)+1-y\right)y_v^*$. After that, when $y_v$ increases from $a$ to $b$, $v$ is charged $\left(a-b+F(a)-F(b)\right)y_v^*$ . Note that the sum of these terms telescopes and is at most $$\left(1-(1-y)+F(1)-F(1-y)\right)y_v^*=\left(y+F(1)-F(1-y)\right)y_v^*.$$

Therefore the total charges to $v$ are at most
\begin{eqnarray*}
&&\left(f(y)+1-y\right)y_v^*+\left(y+F(1)-F(1-y)\right)y_v^*\\
&=& \left(1+f(1-y)+\int_y^1 \frac{1-t}{f(t)}\mathrm{d}t\right)y_v^*\\
&\leq& \beta y_v^*.
\end{eqnarray*}
This implies that the total potential is bounded by $\beta\sum_{v\in V} y_v^*$, which shows that our algorithm is $\beta$-competitive.
\end{proof}

\subsection{Computing the optimal allocation function}
\label{sec:optimization}

The next question is then to find a good $f(y)$ to get a small $\beta$. In essence, the goal is to solve the following optimization problem 
\begin{equation}
\label{eqn:opt}
\inf_{f\in \mathcal{F}}\max_{z\in [0,1]} 1+f(1-z)+\int_z^1 \frac{1-t}{f(t)}\mathrm{d}t.
\end{equation}

where $\mathcal{F}$ is the class of positive continuous functions on $[0,1]$ such that $\frac{1-t}{f(t)}$ is decreasing for each $f\in \mathcal{F}$.

To the best of our knowledge, there is no systematic approach to tackle a minimax optimization problem of this form. A natural way is to first express the optimal $z$ in terms of $f$, and then use techniques from calculus of variation to compute the best $f$. However, a major difficulty is that there is no closed form expression for the optimal $z$. 

To overcome this hurdle, we first disregard the requirement that $\frac{1-t}{f(t)}$ be decreasing. (Though, our final optimal solution turns out to satisfy this condition.) We show that such a relaxation of the optimization problem admits a very nice optimality condition, namely that there exists some optimal $f$ such that $1+f(1-z)+\int_z^1 \frac{1-t}{f(t)}\mathrm{d}t$ is constant for all $z$. We characterize this property in the following lemma.
\begin{lemma}
Let $r:[0,1]\longrightarrow\mathbb{R}_{+}$ be a
continuous function such that for $\forall p\in [0,1]$, $r(p)+\int_{1-p}^{1}\frac{1-x}{r(x)}\mathrm{d}x\leq\gamma$
for some $\gamma>0$. Then there exists a continuous function $f:[0,1]\longrightarrow\mathbb{R}_{+}$
such that $\forall p\in [0,1]$, $f(p)+\int_{1-p}^{1}\frac{1-x}{f(x)}\mathrm{d}x \equiv\gamma$.
\end{lemma}
\begin{proof}
Let $r_{1}=r$ and $R_{1}(p)=r_{1}(p)+\int_{1-p}^{1}\frac{1-x}{r_{1}(x)}\mathrm{d}x$.
Define two sequences of functions $\{r_{i}\},\{R_{i}\}$ recursively
as follows:
\[
r_{i+1}=r_{i}+\gamma-R_{i},R_{i+1}(p)=r_{i+1}(p)+\int_{1-p}^{1}\frac{1-x}{r_{i+1}(x)}\mathrm{d}x.
\]

Note that $r_{i},R_{i}$ are positive and continuous for every $i$. We first show $R_{i}\leq\gamma$ by induction. The base case for $i=1$ is trivial. Now we assume $R_i\leq \gamma$ for some $i$. This implies that $r_i\leq r_{i+1}$.  Then
Notice 
\begin{align*}
R_{i+1}(p)&=r_{i+1}(p)+\int_{1-p}^{1}\frac{1-x}{r_{i+1}(x)}\mathrm{d}x
\leq r_{i+1}(p) +\int_{1-p}^{1}\frac{1-x}{r_{i}(x)}\mathrm{d}x \\
&=r_{i+1}(p) + R_i(p) - r_i(p) = \gamma.
\end{align*}
Therefore $R_i\le \gamma$ for all $i$ and consequently $r_i\leq r_{i+1}$.

Observe that $r_{i}$ converges pointwise as
$r_{i}$ is bounded by $\gamma$ and monotonically increases. Let $r_{\infty}=\lim_{i\rightarrow\infty}r_{i}$. 

Moreover, since $r_{i+1}=r_{i}+\gamma-R_{i}$, $R_{\infty}=\lim_{i\rightarrow\infty}R_{i}\equiv\gamma$.
On the other hand, we have
\begin{align*}
\gamma=R_{\infty}(p) &=\lim_{i\rightarrow\infty}\left(r_{i}(p)+\int_{1-p}^{1}\frac{1-x}{r_{i}(x)}\mathrm{d}x\right)\\
&=r_{\infty}(p)+\lim_{i\rightarrow\infty}\int_{1-p}^{1}\frac{1-x}{r_{i}(x)}\mathrm{d}x.
\end{align*}
By the dominated convergence theorem, $\lim_{i\rightarrow\infty}\int_{1-p}^{1}\frac{1-x}{r_{i}(x)}\mathrm{d}x=\int_{1-p}^{1}\frac{1-x}{r_{\infty}(x)}\mathrm{d}x$
since $\frac{1-x}{r_{i}(x)}$ is bounded by $\frac{1-x}{r_{1}(x)}$. 

By taking limit in the second recurrence, we get
\[
r_{\infty}(p)=\gamma-\int_{1-p}^{1}\frac{1-x}{r_{\infty}(x)}dx
\]
which implies $r_{\infty}$ is  continuous and hence satisfies our requirement.
\end{proof}

Therefore, it is sufficient to consider functions $f$ that satisfy this optimality condition. A consequence is that $f(1-z)=\beta-1-\int_z^1 \frac{1-t}{f(t)}\mathrm{d}t$ is actually differentiable. Differentiating $1+f(1-z)+\int_z^1 \frac{1-t}{f(t)}\mathrm{d}t$ yields $-f'(1-z)-\frac{1-z}{f(z)}=0$, or equivalently, $$f(z)f'(1-z)=z-1.$$

Although this differential equation is atypical as $f(z)$ and $f'(1-z)$ are not taken at the same point, surprisingly it has closed form solutions, as given below.
\begin{lemma}
Let $r$ be a non-negative differentiable function on $[0,1]$ such that $r(z)r'(1-z)=z-1$. Then $$r(z)=\left(\frac{1+k}{2}-z\right)^{\frac{1+k}{2k}}\left(z+\frac{k-1}{2}\right)^{\frac{k-1}{2k}},$$where $k\geq1$. Moreover, $\frac{1-t}{r(t)}$ is decreasing for $t\in [0,1]$.
\end{lemma}
\begin{proof}
We have 
\[
r(p)r'(1-p)=p-1.
\]
Replacing $p$ by $1-p$, we get 
\begin{equation}
\label{eqn:p}
r(1-p)r'(p)=-p.
\end{equation}

Hence,
\begin{equation}
\label{eqn:pp2c}
(r(p)r(1-p))'=1-2p\implies r(p)r(1-p)=p-p^{2}+c
\end{equation}

for some $c$. Note that $r(0)r(1)=c \ge 0$. 
From Eqn~(\ref{eqn:p}) and (\ref{eqn:pp2c}), we get $r'(p)/r(p)=p/(p^{2}-p-c)$.
Let $k=\sqrt{1+4c}\geq 1$. By taking partial fraction and using $(\ln r(p))'=r'(p)/r(p)$,
\[
\frac{r'(p)}{r(p)}=\frac{1}{2k}\left(\frac{1+k}{p-\frac{1+k}{2}}-\frac{1-k}{p-\frac{1-k}{2}}\right)\implies r(p)=D\frac{\left|p-\frac{1+k}{2}\right|^{\frac{1+k}{2k}}}{\left|p-\frac{1-k}{2}\right|^{\frac{1-k}{2k}}}
\]
for some constant $D$. It is easy to check that $r(p)r(1-p)=D^{2}(p-p^{2}-c)\implies D=1$. Since $k\geq 1$, we get the required $r(p)$.

Now we show that
$\frac{1-t}{r(t)}$ is decreasing for $t\in [0,1]$. Taking the derivative of $\frac{1-t}{r(t)}$, we have $-1-(1-t)r'(t)/r(t)=-1-(1-t)t/(t^{2}-t-c)=c/(t^{2}-t-c)\leq 0$,
as desired.
\end{proof}

The final step is just to select the best $f$ from the family of solutions. Since $1+f(1-z)+\int_z^1 \frac{1-t}{f(t)}\mathrm{d}t$ is constant, it suffices to find the smallest $1+f(0)$, which corresponds to the case $k\approx1.1997$, the real fixed point of the hyperbolic cotangent function.\footnote{The optimal $k$ is closely related to the Laplace limit in the solution of Kepler's equation~\cite{weisstein}.}

\begin{theorem}
\label{thm:vcgeneral}
Let $f(z)=\left(\frac{1+k}{2}-z\right)^{\frac{1+k}{2k}}\left(z+\frac{k-1}{2}\right)^{\frac{k-1}{2k}}$, where $k\approx1.1997$. $GreedyAllocation(f)$ is 1.901-competitive for the fractional online vertex cover problem in general graphs.
\end{theorem}

Finally, we remark that our algorithm can be viewed as a generalization of the well-known greedy algorithm because the solution $f(z)=1-z$ (with $k=1$) is equivalent to a variant of the greedy algorithm.

\section{Online fractional matching in general graphs}
We give a primal-dual analysis of the algorithm given in the last section. 
%
A by-product of this primal-dual analysis is a $\frac{1}{1.901}\approx 0.526$-competitive algorithm for online fractional matching in general graphs.

Let $\beta\approx 1.901$ be the competitive ratio established in the last section and $f(z)$ be the same as that of Theorem~\ref{thm:vcgeneral}. Our primal-dual analysis shares some similarities with the one for online bipartite fractional matching by Buchbinder et al.~\cite{Buchbinder2007}.

%
%

Our algorithm $PrimalDual$ applies to both online fractional vertex cover and matching.  When restricted to the dual, it is identical to $GreedyAllocation$. 

\begin{algorithm}[h!]
\SetAlgoLined
\caption{$PrimalDual$}
\label{alg:general greedy}
\KwIn{Online graph $G=(V,E)$}
\KwOut{A fractional vertex cover $\{y_v\}$ of $G$ and a fractional matching $\{x_{uv}\}$.}
Let $T$ be the set of known vertices. Initialize $T=\emptyset$\;
\For{each online vertex $v$}
{
Maximize $y\le 1$, s.t., $\sum_{u\in N(v)\cap T} \max\{y-y_u,0\} \leq f(y)$\;
Let $X = \{u \in N(v)\cap T \,\mid\, y_u < y\}$\;
\For{each $u\in X$}
{
$y_u\leftarrow y$\;
$x_{uv}\longleftarrow \frac{y-y_u}{\beta}\left(1+\frac{1-y}{f(y)}\right)$\;
}
For each $u \in (N(v)\cap T)\setminus X$, $x_{uv}\longleftarrow 0$\;
$y_v \leftarrow 1-y$\;
$T\leftarrow T\cup \{v\}$\;
}
Output $\{y_v\}$ for all $v\in V$\;
\end{algorithm}

To analyze the performance, we claim that the following two invariants hold throughout the execution of the algorithm.

\textbf{Invariant 1:} $$\frac{y_{u}+f(1-z_u)+\int_{z_u}^{y_{u}}\frac{1-t}{f(t)}\mathrm{d}t}{\beta}\geq x_{u},$$ where $z_u$ is the potential of $u$ set upon its arrival, $y_u$ is the current potential of $u$ and $x_u = \sum_{v \in N(u)} x_{uv}$ is the sum of the potentials on the edges incident to $u$. Note that the LHS is at most 1 (see last section for details), which guarantees that the primal is feasible as long as the invariant holds.

\textbf{Invariant 2:} $$\sum_{u\in T} y_{u}=\beta\sum_{(u,v)\in E\cap T^2} x_{uv}$$

Invariant 2 guarantees that the primal and dual objective values are within a factor of $\beta$ from each other. By weak duality, this implies that the algorithm is $\beta$-competitive for online fractional vertex cover and $\frac{1}{\beta}$-competitive for online fractional matching in general graphs.
Note that both invariants trivially hold at the beginning. 

The idea behind Invariant 1 is to enforce some kind of correlation between $y_u$ and $x_u$. For instance, when $y_u$ is small, $x_u$ should not be excessively large because $x_u$ must be increased to (partially) offset any future increase in $y_u$ in order to maintain Invariant 2.


%

We claim that both invariants are preserved.
\begin{lemma}[Invariant 2]
\label{lem:inv2}
In each iteration of the algorithm, the increase in the dual objective value is exactly $\beta$ times that of the primal.
\end{lemma}
\begin{proof}
The dual increment is $$1-y+\sum_{u\in X}(y-y_u)$$ and the primal increment is $$\sum_{u\in X} \frac{y-y_u}{\beta}\left( 1 + \frac{1-y}{f(y)}\right).$$
Thus it suffices to show that $1-y=\sum_{u\in X}(y-y_u)\frac{1-y}{f(y)}$. This just follows from Lemma~\ref{lem:charging}, which states that we have either $y=1$ or $\sum_{u\in X}(y-y_u)=f(y)$.
\end{proof}
\begin{lemma}[Invariant 1]
\label{lem:inv1}
After processing online vertex $v$, we have $x_v\leq \frac{y_v+f(1-y_v)}{\beta}$ and $x_u\leq \frac{y+f(1-z_u)+\int_{z_u}^{y}\frac{1-t}{f(t)}\mathrm{d}t}{\beta}$ for $u\in X$.
\end{lemma}
\begin{proof}
Note that $x_v=\sum_{u\in X}x_{uv}$ is just the increase in the primal objective value. By Invariant 2, $x_v=\frac{1-y+\sum_{u\in X}(y-y_u)}{\beta}$. Our claim for $x_v$ follows since $y_v=1-y$ and $\sum_{u\in X}(y-y_u)\leq f(y)$.

By Invariant 1, the previous $x_u$ satisfies $$x_u-x_{uv}\leq\frac{y_{u}+f(1-z_u)+\int_{z_u}^{y_{u}}\frac{1-t}{f(t)}\mathrm{d}t}{\beta}.$$
This proof is finished by noticing that $$x_{uv}=\frac{y-y_u}{\beta}\left( 1 + \frac{1-y}{f(y)}\right)\leq  \frac{1}{\beta} \left( y-y_u+ \int_{y_u}^y \frac{1-t}{f(t)}\mathrm{d}t\right),$$ as $\frac{1-t}{f(t)}$ is a decreasing function.
\end{proof}

Finally, it is clear that the dual is always feasible. The primal is feasible because $x\geq 0$ and Invariant 1 guarantees that $x_v\leq 1$, as discussed earlier. Combining this and the two lemmas, we have our main result.
\begin{theorem}
\label{thm:pdgeneral}
Our algorithm is $\beta\approx 1.901$-competitive for online fractional vertex cover and $\frac{1}{\beta}\approx 0.526$-competitive for online fractional matching for general graphs.
\end{theorem}

%

It is possible to extend our algorithm to the vertex-weighted fractional vertex cover problem and the fractional b-matching problem, as shown in the appendix.
\begin{theorem}
\label{thm:vertexweighted}
There exists an algorithm that is $\beta\approx 1.901$-competitive for online vertex-weighted fractional vertex cover and $\frac{1}{\beta}\approx 0.526$-competitive for online capacitated fractional matching for general graphs.
\end{theorem}

\section{Hardness Results}

In this section, we obtain new hardness results in our model. All of our hardness results are obtained by considering appropriate bipartite graphs.
Let $G=(L,R,E)$ be a bipartite graph with left vertices $L$ and right vertices $R$. We study different variants of the online vertex cover and matching problems by imposing certain constraints on the vertex arrival order.

\begin{itemize}
\item {\bf 1-alternation.} The left vertices $L$ are offline and the right vertices in $R$ arrive online. When a vertex $v\in R$ arrives, all its incident edges are revealed. This is the case studied in the literature.
\item {\bf $k$-alternation:}  There are $k$ phases and $L_0\subseteq L$ is the set of offline vertices. In each phase $1\leq i\leq k$, if $i$ is odd (resp. even), vertices from a subset of $R$ (resp. $L$) arrive one by one. 
The {\em no alternation} case corresponds to $k =1$.
Note that the case $k=\infty$ effectively removes any constraint on the vertex arrival order, and is called the unbounded alternation case below.
\item {\bf Unbounded alternation:} The vertices in $L\cup R$ arrive in an arbitrary order.
\end{itemize}

\subsection{Lower bounds for the online vertex cover problem}

We give lower bounds on the competitive ratios for online bipartite vertex cover with 2- and 3-alternation, and an upper bound for online bipartite matching with 2-alternation. These hardness results also apply to the more general problems of online vertex cover and matching in general graphs.
\begin{proposition}
There is a lower bound of $1+\frac{1}{\sqrt{2}}\approx 1.707$
for online bipartite vertex cover with 2-alternation.
\end{proposition}
\begin{proof}
It suffices to establish the result for the fractional version of the problem. Suppose that an algorithm $A$ is $(1+\beta)$-competitive. 
Without loss of generality, we may assume that $A$ is deterministic.
Our approach is to bound $\beta$ by considering a family of complete bipartite graphs. Thus a new online vertex is always adjacent to all the vertices on the other side.

Let $|L_0|=d$ and $y$ be the fractional vertex cover maintained by $A$. 
We claim that after processing the $i$-th vertex in $R_1$, we have
$$\sum_{u\in L_0} y_u\leq i\beta.$$ 
The reason is that the adversary can generate infinitely many left online vertices in phase 2 and hence $y_v$, for any $v\in R_1$, converges to 1 (otherwise, if $y_v$, which monotonically increases, converges to some $l<1$, then $y_u\geq 1-l$ for $u\in L_2$ and the cost of the vertex cover found is unbounded while the optimal solution is at most $i$).


Let $v_i^{(1)}$ be the $i$th vertex in $R_1$. Next we claim that $$y_{v_{i}^{(1)}}\geq 1-\frac{i\beta}{d}.$$ Since $\sum_{u\in L_0} y_u\leq i\beta$ after processing $v_{i}^{(1)}$, by Pigeonhole Principle there must be some $y_u\leq \frac{i\beta}{d}$. To maintain a valid vertex cover, we need $y_{v_{i}^{(1)}}\geq 1-\frac{i\beta}{d}$.

Finally, we have 
\begin{equation}
\label{[eqn:yv]}
\sum_{v\in R_1} y_v\leq d\beta.
\end{equation}
Otherwise, the adversary can generate infinitely many online vertices to append $R_1$ in which case 
$y_u$ will be increased to $1$ eventually for all $u\in L_0$, i.e., $\sum_{u\in L_0}y_u = d$.
This contradicts the fact that $A$ is $(1+\beta)$-competitive.

Now by taking $|R_1|=\sqrt{2}d$, we get$$\sum_{i=1}^{\sqrt{2}d}\left( 1-\frac{i\beta}{d}\right) \leq\sum_{v\in R_1} y_v\leq d\beta ,$$
from which our desired result follows by taking $d\longrightarrow\infty$.
\end{proof}

The proofs of the next two results are in the appendix.
\begin{proposition}
There is a lower bound of $1+\sqrt{\frac{1}{2}\left(1+\frac{1}{e^2}\right)}\approx1.753$
for online bipartite vertex cover with 3-alternation.
\end{proposition}
\begin{proof}
Again, let $d=|L_{0}|$. We extend the idea used in the proof of the
bound $1+\frac{1}{\sqrt{2}}$ for 2-alternation. Let $x_{i}$ be the
amount of resources spent on $L_{0}$ by the $i$-th vertex of $R_{1}$, i.e. the increment in the potential of $L_0$.
Let $y_{i}$ be its own potential. Then $y_{i}\geq 1-(x_{1}+\cdots+x_{i})/d$,
$x_{1}+\cdots+x_{i}\leq i\beta$ and $y_{1}+\cdots+y_{i}\leq d\beta$ by the argument used in the proof of Proposition 1.

The new idea is that in phase 2, assuming $|R_1| =i$,  at most $i\beta$ resources can be spent
on $L_{0}$ and $L_{2}$. (This is because the adversary can append infinitely many vertices to the current $L_2$.) Now consider the $j$-th vertex $u_j$ in $L_2$. Similar to Eqn.(\ref{[eqn:yv]}), 
we have $\sum_{v\in R_1} y_v \leq (d+j)\cdot \beta$ after processing $u_j$, since the adversary can append infinitely many online vertices to $R_3$. Consequently, $y_{u_j} \geq 1- \min\{y_v \mid v\in R_1 \} \geq 1- (d+j)\cdot \beta/i$ by the pigeonhole principle. Therefore,

\[
x_{1}+\cdots+x_{i}\leq i\beta-\sum_{j=1}^{\ell}\left(1-\frac{(d+j)\beta}{i}\right),
\]
where $\ell = |L_2|$.


Let $X(i)=x_1+x_2+\cdots+x_i$.
If $i \le d\beta$, we have $X(i) \leq i\beta$. When $i>d\beta$, by setting $\ell = \frac{i}{\beta} - d$, we have 
\begin{align}
X(i) &= i\beta - \left(1-\frac{d\beta}{i}\right)\ell+\frac{\beta}{2i}\ell^2+O(1) \nonumber\\
&= d+i\beta-\frac{i}{2\beta}-\frac{\beta d^2}{2i}+O(1).\nonumber
\label{eqn:Xi}
\end{align}

Notice that our bound on $X(i)$ holds for arbitrary $i$, since the adversary can arbitrarily manipulate the future input graph to fool the deterministic algorithm.

Since $y_i\geq 1-X(i)/d$ and $\sum_{i=1}^k y_i \leq d\beta$ for any $k$, we have

\[
\sum_{i=1}^k \left(1-\frac{X(i)}{d}\right) \leq d\beta.
\]

Let $\alpha = \frac{1}{\sqrt{2\beta^2-1}}$. By setting $k = \beta \alpha d$ and considering $i\leq d\beta$, $i>d\beta$ separately,
we get

\begin{align}
d\beta &\geq \sum_{i=1}^{d\beta} \left(1-\frac{i\beta}{d}\right) +\sum_{i=d\beta +1}^k \left(\frac{i}{2\beta d}+\frac{\beta d}{2i}-\frac{i\beta}{d}+O(1/d)\right)\nonumber
\end{align}

%

By taking $d\longrightarrow\infty$ and using $\sum_{i=1}^n 1/i\approx \ln n$, we have the desired result.

%

\end{proof}

%
%

\subsection{Upper bounds for the online matching problem}
Before establishing our last result on the upper bound for online bipartite matching with 2-alternation, we review how the bound $1-1/e$ is proved for the original problem (i.e. 1-alternation) as the same technique is used in a more complicated way. The next proof is a variant of that in \cite{Karp1990}.
\begin{proposition}
There is an upper bound of $1-1/e\approx 0.632$ for online bipartite matching (with 1-alternation).
\end{proposition}
\begin{proof}
Again, we can consider only the fractional version of the problem and deterministic algorithms. Suppose that an algorithm maintains a fractional matching $x$. Let $L=\{ u_1,...,u_n\}$ and $R=\{ v_1,...,v_n\}$, with $v_i$ adjacent to $u_1,...,u_{n+1-i}$. The size of the maximum matching is clearly $n$. Let $v_1,...,v_n$ be the order in which the online vertices arrive.

Observe that when $v_i$ arrives, $u_1,...,u_{n+1-i}$ are indistinguishable from each other. Thus $x_{v_i}:=\sum_{u\in N(v_i)} x_{uv_i}$ should be evenly distributed to $u_1,...,u_{n+1-i}$, i.e. $x_{uv_i}=\frac{x_{v_i}}{n+1-i}$. This argument can be made formal by considering graphs isomorphic to $G$ with the labels of vertices in $L$ being randomly permuted.

Thus, after processing $v_k$ we have $$x_{u_i}=\frac{x_{v_1}}{n}+...+\frac{x_{v_k}}{n+1-k}.$$ Moreover, the size of the matching found is $x_{v_1}+\cdots+x_{v_n}$ and $x_{v_1},\cdots,x_{v_n}$ satisfy $\frac{x_{v_1}}{n}+\cdots+\frac{x_{v_n}}{1}\leq 1$.

Viewing the above as a LP, it is easy to see that $x_{v_1}+...+x_{v_n}$ is maximized when $x_{v_n},...,x_{v_k}=1,x_{v_{k+1}},...,x_{v_n}=0$ and $\frac{1}{n}+...+\frac{1}{n-k+1}\approx 1$ for some $k$. Now when $n$ is large, $\frac{1}{n}+...+\frac{1}{n-k+1}\approx \ln\frac{n}{n-k}$.

Finally, $$x_{v_1}+\cdots+x_{v_n}=k=n(1-1/e)=(1-1/e) \cdot OPT.$$
\end{proof}

\begin{proposition}
There is an upper bound of 0.6252 for the online matching problem in bipartite graphs with 2-alternation.
\end{proposition}
\begin{proof}
Again, we can consider only the fractional version of the problem and deterministic algorithms. Suppose that an algorithm is $\gamma$-competitive and maintains a fractional matching $x$.

Let $|L_0|=|L_2|=n,|R_1|=2n$. The first $n$ vertices of $R_1$ are adjacent to all vertices in $L_0$. The two subgraphs induced by $L_0$ \& the last $n$ vertices of $R_1$ and $L_1$ \& the first $n$ vertices of $R_1$ are isomorphic to the graph used in the proof of the last theorem. Note that the size of maximum matching is $2n$.

The most important observation here is that after processing the first $n$ vertices of $R_1$, the fractional matching found must have size at least $n\gamma$ as the current optimal solution has size $n$. In other words, we have $$x_{u_{1}^{(1)}}+...+x_{u_{n}^{(1)}}=x_{v_{1}^{(1)}}+...+x_{v_{n}^{(1)}}\geq n\gamma$$after the first $n$ vertices of $R_1$ arrive.

Now the next $n$ vertices of $R_1$, by the same reasoning in the last theorem, are matched to the extent of $k$ such that $\gamma + \frac{1}{n}+...+\frac{1}{n+1-k}\approx 1$, from which we obtain $k=n(1-1/e^{1-\gamma})$. Similarly, $L_2$ is also matched to an extent of $n(1-1/e^{1-\gamma})$.

Putting all the pieces together, we have the inequality $$\frac{n\gamma + 2n(1-1/e^{1-\gamma})}{2n}\geq \gamma\Rightarrow 1-\frac{1}{e^{1-\gamma}}-\frac{\gamma}{2}\geq0.$$

The function $1-\frac{1}{e^{1-\gamma}}-\frac{\gamma}{2}$ is decreasing and has root approximately at 0.6252.
\end{proof}

\section{Discussion and open problems}
%


We presented the first nontrivial algorithm for the online fractional matching and vertex cover problems in graphs where all vertices arrive online. A natural question is whether our competitive ratios, 1.901 and 0.526, are optimal for these two problems. For the special case of the bipartite graphs, can we extend our charging-based framework to get improved algorithms?

Another interesting problem is to beat the greedy algorithm for the online {\em integral} matching problem in bipartite graphs or even general graphs. Very recently, the connection between the optimal algorithms for online bipartite integral and fractional matching was established via the randomized primal-dual method~\cite{devanurrandomized}. This is promising as the techniques developed may also be applicable to our problem. However, it seems quite difficult to reverse-engineer an algorithm for online integral matching based on the analysis of our algorithm.


For online integral vertex cover, as mentioned earlier, there is essentially no hope to do better than 2 assuming the Unique Game Conjecture. Nevertheless, it will still be interesting to obtain an unconditional online hardness result which could be easier than the offline counterpart.

Finally, our discussion has been focused on the {\em oblivious adversary} model. It would be interesting to study our problems in weaker adversary models, i.e., stochastic~\cite{Feldman2009,Manshadi2011} and random arrival models~\cite{Mahdian2011,Karande2011}.


{\noindent \bf Acknowledgments:}  We thank Michel Goemans for helpful discussions and Wang Chi Cheung for comments on a previous draft of this paper.

\bibliography{online}

\begin{thebibliography}{10}

\bibitem{Aggarwal2011}
G.~Aggarwal, G.~Goel, C.~Karande, and A.~Mehta.
\newblock Online vertex-weighted bipartite matching and single-bid budgeted
  allocations.
\newblock In {\em Proceedings of the Twenty-Second Annual ACM-SIAM Symposium on
  Discrete Algorithms}, pages 1253--1264. SIAM, 2011.

\bibitem{Aronson1995}
J.~Aronson, M.~Dyer, A.~Frieze, and S.~Suen.
\newblock Randomized greedy matching. ii.
\newblock {\em Random Structures \& Algorithms}, 6(1):55--73, 1995.

\bibitem{bansal2010lp}
N.~Bansal, A.~Gupta, J.~Li, J.~Mestre, V.~Nagarajan, and A.~Rudra.
\newblock When lp is the cure for your matching woes: improved bounds for
  stochastic matchings.
\newblock {\em Algorithms--ESA 2010}, pages 218--229, 2010.

\bibitem{Blum2006}
A.~Blum, T.~Sandholm, and M.~Zinkevich.
\newblock Online algorithms for market clearing.
\newblock {\em Journal of the ACM (JACM)}, 53(5):845--879, 2006.

\bibitem{Buchbinder2007}
N.~Buchbinder, K.~Jain, and J.~Naor.
\newblock Online primal-dual algorithms for maximizing ad-auctions revenue.
\newblock {\em Algorithms--ESA 2007}, pages 253--264, 2007.

\bibitem{Buchbinder2009}
N.~Buchbinder and J.S. Naor.
\newblock Online primal-dual algorithms for covering and packing.
\newblock {\em Mathematics of Operations Research}, 34(2):270--286, 2009.

\bibitem{NivPersonal}
Niv Buchbinder.
\newblock Personal communication.
\newblock Oct 2012.

\bibitem{Demange2005}
M.~Demange and V.T. Paschos.
\newblock On-line vertex-covering.
\newblock {\em Theoretical Computer Science}, 332(1):83--108, 2005.

\bibitem{devanur2012online}
N.R. Devanur and K.~Jain.
\newblock Online matching with concave returns.
\newblock In {\em Proceedings of the 44th symposium on Theory of Computing},
  pages 137--144. ACM, 2012.

\bibitem{devanurrandomized}
N.R. Devanur, K.~Jain, and R.D. Kleinberg.
\newblock Randomized primal-dual analysis of ranking for online bipartite
  matching.
\newblock In {\em SODA '13: Proceedings of the thirteenth Annual ACM-SIAM
  Symposium on Discrete Algorithms}, 2013.
\newblock to appear.

\bibitem{DevenurH09}
Nikhil~R. Devenur and Thomas~P. Hayes.
\newblock The adwords problem: online keyword matching with budgeted bidders
  under random permutations.
\newblock In {\em EC '09: Proceedings of the tenth ACM conference on Electronic
  commerce}, pages 71--78, New York, NY, USA, 2009. ACM.

\bibitem{Feldman2009}
J.~Feldman, A.~Mehta, V.~Mirrokni, and S.~Muthukrishnan.
\newblock Online stochastic matching: Beating 1-1/e.
\newblock In {\em Foundations of Computer Science, 2009. FOCS'09. 50th Annual
  IEEE Symposium on}, pages 117--126. IEEE, 2009.

\bibitem{goel12}
Goel G. and Tripathi P.
\newblock Matching with our eyes closed.
\newblock In {\em Foundations of Computer Science, 2012. FOCS'12. 53rd Annual
  IEEE Symposium on}. IEEE, 2012.

\bibitem{goel2008online}
G.~Goel and A.~Mehta.
\newblock Online budgeted matching in random input models with applications to
  adwords.
\newblock In {\em SODA}, volume~8, pages 982--991, 2008.

\bibitem{kalyanasundaram2000optimal}
B.~Kalyanasundaram and K.R. Pruhs.
\newblock An optimal deterministic algorithm for online b-matching.
\newblock {\em Theoretical Computer Science}, 233(1):319--325, 2000.

\bibitem{KP00}
Bala Kalyanasundaram and Kirk Pruhs.
\newblock Speed is as powerful as clairvoyance.
\newblock {\em J. ACM}, 47(4):617--643, 2000.

\bibitem{Karande2011}
C.~Karande, A.~Mehta, and P.~Tripathi.
\newblock Online bipartite matching with unknown distributions.
\newblock In {\em Proceedings of the 43rd annual ACM symposium on Theory of
  computing}, pages 587--596. ACM, 2011.

\bibitem{Karlin2001}
A.R. Karlin, C.~Kenyon, and D.~Randall.
\newblock Dynamic tcp acknowledgement and other stories about e/(e-1).
\newblock In {\em Proceedings of the thirty-third annual ACM symposium on
  Theory of computing}, pages 502--509. ACM, 2001.

\bibitem{Karlin1994}
A.R. Karlin, M.S. Manasse, L.A. McGeoch, and S.~Owicki.
\newblock Competitive randomized algorithms for nonuniform problems.
\newblock {\em Algorithmica}, 11(6):542--571, 1994.

\bibitem{karlin1988competitive}
A.R. Karlin, M.S. Manasse, L.~Rudolph, and D.D. Sleator.
\newblock Competitive snoopy caching.
\newblock {\em Algorithmica}, 3(1):79--119, 1988.

\bibitem{Karp1990}
R.M. Karp, U.V. Vazirani, and V.V. Vazirani.
\newblock An optimal algorithm for on-line bipartite matching.
\newblock In {\em Proceedings of the twenty-second annual ACM symposium on
  Theory of computing}, pages 352--358. ACM, 1990.

\bibitem{Khot2008}
S.~Khot and O.~Regev.
\newblock Vertex cover might be hard to approximate to within 2-[epsilon].
\newblock {\em Journal of Computer and System Sciences}, 74(3):335--349, 2008.

\bibitem{Lotker2008}
Z.~Lotker, B.~Patt-Shamir, D.~Rawitz, and S.~Albers.
\newblock Rent, lease or buy: Randomized algorithms for multislope ski rental.
\newblock In {\em 25th International Symposium on Theoretical Aspects of
  Computer Science (STACS 2008)}, volume~1, pages 503--514, 2008.

\bibitem{poloczek12}
Poloczek M. and Szegedy M.
\newblock Randomized greedy algorithms for the maximum matching problem with
  new analysis.
\newblock In {\em Foundations of Computer Science, 2012. FOCS'12. 53rd Annual
  IEEE Symposium on}. IEEE, 2012.

\bibitem{Mahdian2011}
M.~Mahdian and Q.~Yan.
\newblock Online bipartite matching with random arrivals: an approach based on
  strongly factor-revealing lps.
\newblock In {\em Proceedings of the 43rd annual ACM symposium on Theory of
  computing}, pages 597--606. ACM, 2011.

\bibitem{Manshadi2011}
V.H. Manshadi, S.O. Gharan, and A.~Saberi.
\newblock Online stochastic matching: Online actions based on offline
  statistics.
\newblock In {\em Proceedings of the Twenty-Second Annual ACM-SIAM Symposium on
  Discrete Algorithms}, pages 1285--1294. SIAM, 2011.

\bibitem{Mehta2007}
A.~Mehta, A.~Saberi, U.~Vazirani, and V.~Vazirani.
\newblock Adwords and generalized online matching.
\newblock {\em Journal of the ACM (JACM)}, 54(5):22, 2007.

\bibitem{weisstein}
Eric~W. Weisstein.
\newblock Hyperbolic cotangent. {From MathWorld---A Wolfram Web Resource}.

\end{thebibliography}
\bibliographystyle{plain}

\appendix

\section{Reduction from Multislope Ski Rental to Online weighted Bipartite Vertex Cover}
\label{sec:multislope}
There is a total of $n$ states $[n]$ in the multislope ski rental problem. Each state $i$ associated with buying cost $b_i$ and rental cost $r_i$. As argued in [.], we may assume that we start in state 1 and have $0=b_1\leq b_2\leq \ldots\leq b_n, r_1\geq r_2\geq \ldots\geq r_n\geq 0$. The game starts at time 0 and ends at some unknown time $t_{end}$ determined by the adversary.

At each time $t\in [0,t_{end}]$, we can transition from the current state $i$ to some state $j>i$. Let state $f$ be the final state at time $t_{end}$. The total cost incurred is given by $$b_f+\sum_{i=1}^f x_i r_i,$$where $x_i$ is the amount of time spent in state $i$. The classical ski rental problem corresponds to $n=2$ and $b_1=0,b_2=B,r_1=1,r_2=0$.

Consider now the discrete version of this problem. We discretize time into consecutive intervals of length $\epsilon$ for some small $\epsilon >0$. At the beginning of each interval, we can stay in the current state $i$ or transition from $i$ to some state $j>i$. Each of the two choices correspond to a cost of $r_i\epsilon$ or $b_j-b_i+r_j\epsilon$.


We are ready to describe the reduction to online vertex-weighted bipartite vertex cover. Let $L=\{1,2,\ldots,n\}$ with weights $w_i=b_{i+1}-b_i$ for $i<n$ and $w_n=\infty$. The $(qn+r)$-th online vertex $v_{qn+k}\in R$, where $q$ is a nonnegative integer and $1\leq k\leq n$, has weight $(r_{k}-r_{k+1})\epsilon$ (with $r_{n+1}=0$) and is adjacent to the left vertices $1,\ldots,k$.

Intuitively, the $(q+1)$-th time interval is represented by the online vertices $qn+1,\ldots,qn+n$. If we are in state $i$, (1) the left vertices $1,\ldots,i-1$ should be covered and have total weight $b_i-b_0=b_i$ and, (2) the online vertices $qn+i\ldots qn+n$ should be covered and have total weight $r_i\epsilon$. Thus when we transition from state $i$ to state $j>i$, the vertices $i,\ldots,j-1$ should be added to the cover. Moreover, the left vertex $n$, which has infinite weight, is used to ensure that the algorithm is forced to put the online vertex $qn+n$, which has weight $r_n$, into the cover.

Finally, we show that a $c$-competitive algorithm for online vertex-weighted bipartite vertex cover gives a $c$-competitive algorithm for multislope ski rental under the above reduction. Consider the vertex cover maintained by the algorithm after processing online vertices $qn+1,\ldots,qn+n$. Suppose that $1,\ldots,i-1$ are in the cover but $i$ is not. Then the online vertices  $qn+i\ldots qn+n$ must also be in the cover. Thus we can simply stay in (or transition to if the previous state is smaller) state $i$. It is clear that this strategy is valid by the preceding discussion. Furthermore, the cost incurred by the algorithm for multislope ski rental is no greater than the counterpart for vertex cover.

\section{Proof of Theorem 2}
We modify GreedyAllocation as follows. The only difference is that $\sum_{u\in N(v)\cap T} \max\{y-y_u,0\} \leq f(y)$ is replaced by $\sum_{u\in N(v)\cap T} w_u\max\{y-y_u,0\} \leq w_vf(y)$.

\begin{algorithm}[h!]
\SetAlgoLined
\caption{$GreedyAllocation$ with allocation function $y+\alpha$}
\label{alg:general greedy}
\KwIn{Online graph $G=(V,E)$ with offline vertices $U\subset V$}
\KwOut{A fractional vertex cover of $G$}
Initialize for each $u\in U$, $y_u = 0$\;
Let $T$ be the set of known vertices. Initialize $T=U$\;
\For{each online vertex $v$}
{
Maximize $y\le 1$, s.t., $\sum_{u\in N(v)\cap T} w_u\max\{y-y_u,0\} \leq w_v(y+\alpha)$\;
For each $u\in N(v)\cap T$, $y_u \leftarrow \max\{ y_u, y\}$\;
$y_v \leftarrow 1-y$\;
$T\leftarrow T\cup \{v\}$\;
}
Output $\{y_v\}$ for all $v\in V$\;
\end{algorithm}

To analyze the algorithm, we need the following lemma which is an easy extension of Lemma~\ref{lem:charging}.

\begin{lemma}
\label{lem:weightedcharging}
Let $f:[0,1]\longrightarrow\mathbb{R}_+$ be continuous such that $\frac{1-t}{f(t)}$ is decreasing, and $F(x)=\int_0^x  \frac{1-t}{f(t)}\mathrm{d}t$. If $\sum_{u\in X} w_u(y-y_u)= w_vf(y)$ for some set $X$ and $y\geq y_u$ for $u\in X$, then
$$w_v(1-y)\leq \sum_{u\in X} w_u\left(F(y)-F(y_u)\right).$$
\end{lemma}

We only give a sketch of the charging scheme as it is very similar to the unweighted case.

We charge the potentials used to the vertices of the minimum cover $C^*$. Let $v$ be an online vertex. The case $v\in C^*$ is trivial.

Now consider the case $v\notin C^*$. We charge the potential spent on $u\in N(v)\subseteq C^*$ to $u$ itself. The potential spent on $v$ is $y_v = w_v(1-y)$ where $y$ is the final water level. 
Let $X\subset N(v)$ be the set of vertices whose potentials increase when processing $v$.
If $y =1$, we are done. If $y <1$, we have $\sum_{u\in X} w_u(y-y_u) = w_v(\alpha +y)$, where $y_u$ is the potential of $u$ before processing $v$. By Lemma~\ref{lem:weightedcharging}, $w_v(1-y) \leq \sum_{u\in X} w_u(G(y)-G(y_u))$.

Now the charges to each $u\in C^*\cap R$ are at most $1+G(1)-G(0)=1+\alpha$.

\section{ Proof of Theorem 6}
The proof is an extension of the one for Theorem~\ref{thm:pdgeneral}. We analyze the following algorithm using the primal-dual method. It is also possible to give a charging-based analysis of the online vertex-weighted vertex cover part of the algorithm.

The function $f$ below is the same as that for Theorem~\ref{thm:pdgeneral}. Recall that $\beta\geq 1+f(1-z)+\int_z^1 \frac{1-t}{f(t)}\mathrm{d}t$ for $z\in [0,1]$.

\begin{algorithm}[h!]
\SetAlgoLined
\caption{$PrimalDual-Weighted$}
\label{alg:general greedy weighted}
\KwIn{Online graph $G=(V,E)$, weights/capacities $w_v,v\in V$}
\KwOut{A fractional vertex cover $\{y_v\}$ of $G$ and a fractional capacitated matching $\{x_{uv}\}$.}
Let $T$ be the set of known vertices. Initialize $T=\emptyset$\;
\For{each online vertex $v$}
{
Maximize $y\le 1$, s.t., $\sum_{u\in N(v)\cap T} w_u\max\{y-y_u,0\} \leq w_vf(y)$\;
Let $X = \{u \in N(v)\cap T \,\mid\, y_u < y\}$\;
\For{each $u\in X$}
{
$y_u\leftarrow y$\;
$x_{uv}\longleftarrow \frac{w_u(y-y_u)}{\beta}\left(1+\frac{1-y}{f(y)}\right)$\;
}
For each $u \in (N(v)\cap T)\setminus X$, $x_{uv}\longleftarrow 0$\;
$y_v \leftarrow 1-y$\;
$T\leftarrow T\cup \{v\}$\;
}
Output $\{y_v\}$ for all $v\in V$\;
\end{algorithm}

We claim that the following two invariants hold:

\textbf{Invariant 1:} $$w_u\cdot\frac{y_{u}+f(1-z_u)+\int_{z_u}^{y_{u}}\frac{1-t}{f(t)}\mathrm{d}t}{\beta}\geq x_{u},$$ where $z_u$ is the potential of $u$ set upon its arrival, $y_u$ is the current potential of $u$ and $x_u = \sum_{v \in N(u)} x_{uv}$ is the sum of the potentials on the edges incident to $u$. Note that the LHS is at most 1  by the definition of $\beta$, which guarantees that the primal is feasible as long as the invariant holds.

\textbf{Invariant 2:} $$\sum_{u\in T} w_uy_{u}=\beta\sum_{(u,v)\in E\cap (T\times T)} x_{uv}$$

We sketch why these two invariants are preserved after processing each vertex $v$. The proof is almost identical to the unweighted case.

Invariant 2: The dual increment is $$w_v(1-y)+\sum_{u\in X}w_u(y-y_u)$$ and the primal increment is $$\sum_{u\in X} \frac{w_u(y-y_u)}{\beta}\left( 1 + \frac{1-y}{f(y)}\right).$$
Thus it suffices to show that $w_v(1-y)=\sum_{u\in X}w_u(y-y_u)\frac{1-y}{f(y)}$. 
When $y=1$, the statement trivially holds. On the other hand, when $y<1$, by construction, we have
$\sum_{u\in X}w_u(y-y_u)=w_vf(y)$.

Invariant 1: We first show that Invariant 1 still holds for $x_v$. Note that $x_v=\sum_{u\in X}x_{uv}$ is just the increase in the primal objective value. By Invariant 2, we have
$$x_v=\frac{w_v(1-y)+\sum_{u\in X}w_u(y-y_u)}{\beta}\le \frac{w_v}{\beta}(z_u + f(1-z_u)).$$

We now show that Invariant 1 is preserved for each $u\in V$. By Invariant 1, the previous $x_u$ satisfies $$x_u-x_{uv}\leq w_u\frac{y_{u}+f(1-z_u)+\int_{z_u}^{y_{u}}\frac{1-t}{f(t)}\mathrm{d}t}{\beta}.$$
This proof is finished by noticing that $$x_{uv}=\frac{w_u(y-y_u)}{\beta}\left( 1 + \frac{1-y}{f(y)}\right)\leq  \frac{w_u}{\beta} \left( y-y_u+ \int_{y_u}^y \frac{1-t}{f(t)}
\mathrm{d}t\right),$$ as $\frac{1-t}{f(t)}$ is a decreasing function.

\end{document}